\documentclass[11pt]{article}

\usepackage{amsmath}
\usepackage{amssymb}
\usepackage{latexsym}
\usepackage{epsfig}
\usepackage{array}
\usepackage{color}
\usepackage{enumerate}
\usepackage{fullpage}
\usepackage[margin=1in]{geometry}
\usepackage{xspace}
\usepackage{times}

\advance\oddsidemargin by -1in  
\advance\evensidemargin by -1in 
\advance\topmargin by -1in      
\advance\topmargin by -37pt     
\skip\footins 8.1pt plus 4pt minus 2pt  
\newtheorem{definition}{Definition}
\newtheorem{observation}{Observation}
\newtheorem{theorem}{Theorem}
\newtheorem{lemma}{Lemma}
\newtheorem{claim}{Claim}

\newcommand{\ts}{t^\star}
\newcommand{\ra}{\frac{\alpha-\beta}{\alpha-1}}
\newcommand{\rainv}{\frac{\alpha-1}{\alpha-\beta}}
\newcommand{\bs}{{\ensuremath {\hat{a}}}\xspace}
\newcommand{\expo}{1-\delta}

\newlength{\boxwidth}
\makeatletter
\DeclareRobustCommand{\qed}{%
  \ifmmode 
  \else \leavevmode\unskip\penalty9999 \hbox{}\nobreak\hfill
  \fi
  \quad\hbox{\qedsymbol}}
\newcommand{\openbox}{\leavevmode
  \hbox to.77778em{%
  \hfil\vrule \vbox to.675em{\hrule width.6em\vfil\hrule}%
  \vrule\hfil}}
\newcommand{\qedsymbol}{\openbox}

\newenvironment{proof}[1][\proofname]{\par \normalfont
  \topsep6\p@\@plus6\p@ \trivlist
  \item[\hskip\labelsep\bfseries\itshape #1.]\ignorespaces }{%
  \qed\endtrivlist }

\newcommand{\proofname}{Proof}

\newcommand{\hide}[1]{}

\begin{document}

\title{On the Approximation Performance of Fictitious Play in Finite Games
\thanks{Supported by EPSRC grants EP/G069239/1 and EP/G069034/1 ``Efficient Decentralised Approaches in Algorithmic
Game Theory.''}}

\author{Paul W. Goldberg\thanks{Department of Computer Science, University of Liverpool, Ashton Street, Liverpool L69 3BX, U.\,K.} \and Rahul Savani\footnotemark[2] \and Troels Bjerre S\o rensen\thanks{Department of Computer Science,
University of Warwick, Gibbet Hill Rd,
Coventry CV4 7AL, U.\ K.} \and Carmine Ventre\footnotemark[2]}
\date{}
\maketitle

\begin{abstract}
We study the performance of Fictitious Play, when used as a heuristic
for finding an approximate Nash equilibrium of a 2-player game.
We exhibit a class of 2-player games having payoffs in the range $[0,1]$
that show that Fictitious Play fails to find a solution having an additive
approximation guarantee significantly better than $1/2$.
Our construction shows that for $n\times n$ games, in the worst case both players may
perpetually have mixed strategies whose payoffs fall short of the best response by
an additive quantity $1/2 - O(1/n^{1-\delta})$ for arbitrarily small $\delta$.
We also show an essentially matching upper bound of $1/2 - O(1/n)$.
\end{abstract}

\section{Introduction}
Fictitious Play is a very simple iterative process for computing
equilibria of games.
A detailed motivation for it is given in~\cite{Con}.
When it converges, it necessarily converges to a
Nash equilibrium. For 2-player games, it is known to converge for
zero-sum games \cite{Robinson}, 
or if one player has just 2 strategies~\cite{Ber05}. On the other
hand, Shapley exhibited a $3\times 3$ game for which it fails to converge \cite{FL,Shapley}.

Fictitious Play (FP) works as follows. Suppose that each player has a number of {\em actions},
or pure strategies. Initially (at iteration 1) each player starts with a single action.
Thereafter, at iteration $t$, each player has a sequence of $t-1$ actions which
is extended with a $t$-th action chosen as follows. Each player makes a best
response to a distribution consisting of the selection of an opponent's strategy
uniformly at random from his sequence. (To make the process precise, a tie-breaking
rule should also be specified; however, in the games constructed here, there will be no ties.)
Thus the process generates a sequence of mixed-strategy profiles (viewing the sequences
as probability distributions), and the hope is that they converge to a limiting
distribution, which would necessarily be a Nash equilibrium.

The problem of computing {\em approximate} equilibria was motivated by the
apparent intrinsic hardness of computing exact equilibria~\cite{DGP}, even in the 2-player case~\cite{CDT}.
An $\epsilon$-Nash equilibrium is one where each player's strategy has a payoff
of at most $\epsilon$ less than the best response.
Formally, for 2 players with pure strategy sets $M$, $N$ and
payoff functions $u_i : M \times N \to \mathbb{R}$ for $i \in \{1, 2\}$,
the mixed strategy $\sigma$ is an
\emph{$\epsilon$-best-response} against the mixed strategy $\tau$, if for any $m\in M$, we have
$u_1(\sigma,\tau)\ge u_1(m,\tau)-\epsilon.$
A pair of strategies $\sigma$, $\tau$ is an $\epsilon$-Nash equilibrium if they are
$\epsilon$-best responses to each other.
Typically one assumes that
the payoffs of a game are rescaled to lie in $[0,1]$, and then a general question is:
for what values of $\epsilon$ does some proposed algorithm guarantee to find
$\epsilon$-Nash equilibria? Previously, the focus has been on various
algorithms that run in polynomial time. Our result for FP applies without any limit
on the number of iterations; we show that a kind of cyclical behavior persists.

A recent paper of Conitzer~\cite{Con} shows that FP obtains an approximation
guarantee of $\epsilon = (t+1)/2t$ for 2-player games, where $t$ is the number of
FP iterations, and furthermore, if both players have
access to infinitely many strategies, then FP cannot do better than this.
The intuition behind this upper bound 
is that an action that appears
most recently in a player's sequence has an $\epsilon$-value close to 0 (at most~$1/t$);
generally a strategy that occurs a fraction $\gamma$ back in the sequence has
an $\epsilon$-value of at most slightly more than $\gamma$
(it is a best response to slightly less than $1-\gamma$ of the
opponent's distribution), and the $\epsilon$-value of a player's mixed
strategy is at most the overall average, i.e., $(t+1)/2t$, which approaches $1/2$ as
$t$ increases.

However, as soon as the number of available pure strategies is exceeded by the number
of iterations of FP, various pure strategies must get re-used, and this re-usage
means, for example, that every previous occurrence of the most recent action all
have $\epsilon$-values of $1/t$. This appears to hold out the hope that FP may
ultimately guarantee a significantly better additive approximation. We show that unfortunately
that is not what results in the worst case.
Our hope is that this result may either guide the design of more ``intelligent''
dynamics having a better approximation performance, or alternatively generalize
to a wider class of related algorithms, for example the ones discussed in~\cite{DFPPV}.

In Section~\ref{sec:LB} we give our main result, the lower bound 
of $1/2-O(1/n^{1-\delta})$ for any $\delta>0$,
and in Section~\ref{sec:UB} we give the corresponding upper bound of $1/2-O(1/n)$.

\section{Lower Bound}\label{sec:LB}


\begin{figure*}[htb]
\centering
\begin{tabular}{ *{5}{|c} | *{5}{|c} | *{10}{|c} | }
\hline
\hline
  &   &   &   &   &   &   &   &   &   &   &   &   &   &   &   &   &   & &  \\ \hline
1 &   &   &   &   &   &   &   &   &   &   &   &   &   &   &   &   &   &  & \\ \hline
$\beta$ & 1 &   &   &   &   &   &   &   &   &   &   &   &   &   &   &   &   &&   \\ \hline
$\beta$ & $\beta$ & 1 &   &   &   &   &   &   &   &   &   &   &   &   &   &   &   &  & \\ \hline
$\beta$ & $\beta$ & $\beta$ & 1 &   &   &   &   &   &   &   &   &   &   &   &   &   &   & &  \\ \hline \hline
$\beta$ & $\beta$ & $\beta$ & $\beta$ & $\alpha$ & 1 &   &   &   &   &   &   &   &   &   &   &   &   & &  \\ \hline
$\beta$ & $\beta$ & $\beta$ & $\beta$ & $\beta$ & $\alpha$ & 1 &   &   &   &   &   &   &   &   &   &   &   &  & \\ \hline
$\beta$ & $\beta$ & $\beta$ & $\beta$ & $\beta$ & $\beta$ & $\alpha$ & 1 &   &   &   &   &   &   &   &   &   &   &  & \\ \hline
$\beta$ & $\beta$ & $\beta$ & $\beta$ & $\beta$ & $\beta$ & $\beta$ & $\alpha$ & 1 &   &   &   &   &   &   &   &   &   &  & \\ \hline
$\beta$ & $\beta$ & $\beta$ & $\beta$ & $\beta$ & $\beta$ & $\beta$ & $\beta$ & $\alpha$ & 1 &   &   &   &   &   &   &   &   & &  \\ \hline \hline
$\beta$ & $\beta$ & $\beta$ & $\beta$ & $\beta$ & $\beta$ & $\beta$ & $\beta$ & $\beta$ & $\alpha$ & 1 &   &   &   &  &$\beta$ & $\beta$ & $\beta$ & $\beta$ & $\alpha$ \\ \hline
$\beta$ & $\beta$ & $\beta$ & $\beta$ & $\beta$  & $\beta$ &  $\beta$ & $\beta$ & $\beta$ & $\beta$ & $\alpha$ & 1 &   &   &  &  & $\beta$ &$\beta$ & $\beta$ & $\beta$ \\ \hline
$\beta$ & $\beta$ & $\beta$ & $\beta$ & $\beta$  & $\beta$  & $\beta$  & $\beta$ & $\beta$ & $\beta$ & $\beta$ & $\alpha$ & 1 &   &  &  &  & $\beta$ &$\beta$ & $\beta$ \\ \hline
$\beta$ & $\beta$ & $\beta$ & $\beta$ & $\beta$   &  $\beta$ &  $\beta$ & $\beta$  & $\beta$ & $\beta$ & $\beta$ & $\beta$ & $\alpha$ & 1 &  &  &  &  & $\beta$ &$\beta$ \\ \hline
$\beta$ & $\beta$ & $\beta$ & $\beta$ & $\beta$   &  $\beta$ &  $\beta$ & $\beta$  &  $\beta$ & $\beta$ & $\beta$ & $\beta$ & $\beta$ & $\alpha$ & 1 &  &  &  & & $\beta$\\ \hline
$\beta$ & $\beta$ & $\beta$ & $\beta$ & $\beta$  & $\beta$  &  $\beta$ & $\beta$  & $\beta$  & $\beta$  & $\beta$ & $\beta$ & $\beta$ & $\beta$ & $\alpha$ & 1 &  &  & &\\ \hline
$\beta$ & $\beta$ & $\beta$ & $\beta$ & $\beta$  & $\beta$  & $\beta$  & $\beta$  & $\beta$  &  $\beta$ &   & $\beta$ & $\beta$ & $\beta$ & $\beta$ & $\alpha$ & 1 &  & &\\ \hline
$\beta$ & $\beta$ & $\beta$ & $\beta$ & $\beta$   &  $\beta$ & $\beta$  & $\beta$  & $\beta$  & $\beta$  &   &   & $\beta$ & $\beta$ & $\beta$ & $\beta$ & $\alpha$ & 1 & &\\ \hline
$\beta$ & $\beta$ & $\beta$ & $\beta$ & $\beta$   &   $\beta$ & $\beta$ & $\beta$ & $\beta$ & $\beta$   &   &   &   & $\beta$ & $\beta$ & $\beta$ & $\beta$ & $\alpha$ & 1 &\\ \hline
$\beta$ & $\beta$ & $\beta$ & $\beta$ & $\beta$   &   $\beta$ & $\beta$ & $\beta$ & $\beta$ & $\beta$   &   &   &   &  & $\beta$ & $\beta$ & $\beta$ & $\beta$ & $\alpha$ & 1\\ \hline
\end{tabular}
\caption{The game ${\cal G}_5$ belonging to the class of games used to prove the lower bound.}\label{fig:LBgame}
\end{figure*}


We specify a class of games with parameter $n$, whose
general idea is conveyed in
Figure \ref{fig:LBgame}, which shows the row player's matrix
for $n=5$; the column player's matrix is its transpose. A blank entry
indicates the value zero; let $\alpha=1+\frac{1}{n^{\expo}}$ and
$\beta=1-\frac{1}{n^{2(\expo)}}$ for $\delta > 0$. Both players
start at strategy 1 (top left).
Generally, let ${\cal G}_n$ be a
$4n\times 4n$ game in which the column player's payoff matrix is the
transpose of the row player's payoff matrix $R$, which itself is specified
as follows. For $i,j\in [4n]$ we have
\begin{itemize}
\setlength{\itemsep}{0pt}
\setlength{\parskip}{0pt}
\setlength{\parsep}{0pt}
\item If $i\in [2:n]$, $R_{i,i-1}=1$. If $i\in [n+1:4n]$, $R_{i,i}=1$.
\item If $i\in [n+1:4n]$, $R_{i,i-1}=\alpha$. Also, $R_{2n+1,4n}=\alpha$.
\item Otherwise, if $i>j$ and $j\leq 2n$, $R_{i,j}=\beta$.
\item Otherwise, if $i>j$ and $i-j\leq n$, $R_{i,j}=\beta$.
If $j\in[3n+1:4n]$, $i\in [2n+1:j-n]$, $R_{i,j}=\beta$.
\item Otherwise, $R_{ij}=0$.
\end{itemize}
For ease of presentation we analyze FP on ${\cal G}_n$; the obtained results can be
seen to apply to a version of ${\cal G}_n$ with payoffs rescaled into $[0,1]$ (cf. the proof of Theorem \ref{thm:final}).

\subsection{Overview}

In this section we give a general overview and intuition on how our main
result works, before embarking on the technical details. Number the
strategies $1,\ldots 4n$ from top to bottom and left to right, and
assume that both players start at strategy 1. Fictitious Play proceeds
in a sequence of steps, which we index by positive integer $t$, so that
step $t$ consists of both players adding the $t$-th element to their
sequences of length $t-1$. We have the following observation:
\begin{observation}\label{obs:sym}
Since the column player's payoff matrix is the transpose of the row player's,
at every step both players play the same action.
\end{observation}
This simplifies the analysis since it means we are analyzing a single
sequence of numbers (the shared indices of the actions chosen by the players).

A basic insight into the behavior of Fictitious Play on the games in
question is provided by Lemma~\ref{le:bestresponses}, which tells us a
great deal about the structure of the players' sequence. Let $s_t$ be
the action played at step $t$. We set $s_1=1$.

\begin{lemma}\label{le:bestresponses}
For any time step $t$, if $s_t \neq s_{t+1}$ then $s_{t+1} = s_t+1$  (or $s_{t+1} = 2n+1$ if $s_t = 4n$).
\end{lemma}
\begin{proof}
The first $n$ steps are similar to~\cite{Con}.
For step $t>n$, suppose the players play $s_t \neq 4n$ (by Observation \ref{obs:sym}, the two
players play the same strategy).  $s_t$ is a best response at step $t$, and
since $R_{s_t+1,s_t}>R_{s_t,s_t}>R_{j,s_t}$ ($j\not\in\{s_t,s_t+1\}$), strategy
$s_t+1$ is the only other candidate to become a better response after $s_t$ is
played.  Thus, if $s_{t+1}\not=s_t$, then $s_{t+1}=s_t+1$. Similar arguments
apply to the case $s_t=4n$.
%
%
\end{proof}

The lemma implies that the sequence consists of a block of consecutive 1's
followed by some consecutive 2's, and so on through all the actions in
ascending order until we get to a block of consecutive $4n$'s. The
blocks of consecutive actions then cycle through the actions
$\{2n+1,\ldots,4n\}$ in order, and continue to do so repeatedly.

As it stands, the lemma makes no promise about the lengths of these
blocks, and indeed it does not itself rule out the possibility that one
of these blocks is infinitely long (which would end the cycling process
described above and cause FP to converge to a pure Nash equilibrium).
The subsequent results say more about the lengths of the blocks. They show
that in fact the process never converges (it cycles infinitely often)
and furthermore, the lengths of the blocks increase in geometric
progression. The parameters $\alpha$ and $\beta$ in ${\cal G}_n$ govern the
ratio between the lengths of consecutive blocks. We choose a ratio large
enough that ensures that the $n$ strategies most recently played, occupy
all but an exponentially-small fraction of the sequence. At the same
time the ratio is small enough that the corresponding probability distribution
does not allocate much probability to any individual strategy.

As an aside, we conclude with the following observation, which is not hard to check from
the structure of the game.
\begin{observation}
The game has a mixed Nash equilibrium in which both players use the uniform
distribution over strategies $\{2n+1,\ldots,4n\}$. The equilibrium has payoff
approximately $\frac{1}{2}$ to each player. There are no pure Nash equilibria,
although if both players use the same strategy in the range $\{n+1,\ldots,4n\}$
then they would receive payoff 1. Recall that $\alpha>1$, so a payoff of 1 to each player does not imply an
equilibrium.
\end{observation}

\subsection{The proof}

We now identify some properties of probabilities assigned to strategies by FP. We let $\ell_t(i)$ be the number of times that strategy $i$ is played by the players until time step $t$ of FP.
Let $p_t(i)$ be the corresponding probability assigned by the players to strategy $i$ at step $t$, also for any subset of actions $S$ we use $p_t(S)$ to denote the total probability of elements of $S$.
So it is immediate to observe that $$p_t(i)=\frac{\ell_t(i)}{\sum_{j=1}^{4n} \ell_t(j)}=\frac{\ell_t(i)}{t}.$$ The next fact easily follows from the FP rule.

\begin{lemma}\label{lemma:p1}
For all strategies $i \leq n$, $p_t(i)=\frac{1}{t}$ and therefore $\ell_t(i)=1$ for any time step $t\geq i$.
\end{lemma}
\begin{proof}
At step 1, each player sets $p_1(1)=1$ and $p_1(i)=0$ for $i>1$.
As in~\cite{Con}, for $t\leq n$ the sequence chosen by both players is $(1,2,\ldots,t)$,
so $p_t(i)=\frac{1}{i}$ for $i\leq t$ and $0$ otherwise. Lemma \ref{le:bestresponses} implies that none of the first $n$ strategies will be a best response subsequently, thus implying the claim.
\end{proof}

By Lemma \ref{le:bestresponses}, each strategy is played a number of consecutive times, in order, until the strategy $4n$ is played; at this point, this same pattern repeats but only for the strategies in $\{2n+1,\ldots, 4n\}$. We let $\ts$ be the length of the longest sequence containing all the strategies in ascending order, that is, $\ts$ is the last step of the first consecutive block of $4n$'s.
%
We also let $t_i$ be the last time step in which $i$ is played during the first $\ts$ steps,
i.e., $t_i$ is such that $\ell_{t_i}(i)=\ell_{t_i-1}(i)+1$ and $\ell_{t}(i)=\ell_{\ts}(i)$ for $t\in \{t_i,\ldots,\ts\}$.

\begin{lemma}\label{lemma:p2}
For all strategies $n+1 \leq i \leq 3n$ and all $t \in \{t_i,\ldots,\ts\}$, it holds:
\[
\ra p_t(i-1) \leq p_t(i) \leq \frac{1}{t} + \ra p_t(i-1)
\]
and therefore,
\[
\ra \ell_t(i-1) \leq \ell_t(i) \leq {1} + \ra \ell_t(i-1).
\]
\end{lemma}
\begin{proof}
By definition of $t_i$, strategy $i$ is played at step $t_i$. This means that $i$ is a best response for the players given the probability distributions at step $t_i-1$. In particular, the expected payoff of $i$ is better than the expected payoff of $i+1$, that is,
\begin{align*}
& \beta \sum_{j=1}^{i-2} p_{t_i-1}(j) + \alpha p_{t_i-1}(i-1)+p_{t_i-1}(i) \geq \\ & \beta \sum_{j=1}^{i-2} p_{t_i-1}(j) + \beta p_{t_i-1}(i-1)+\alpha p_{t_i-1}(i).
\end{align*}
Since $\alpha>1$, the above implies that $p_{t_i-1}(i) \leq \ra p_{t_i-1}(i-1)$.
By explicitly writing the probabilities, we get
\begin{align}
& \frac{\ell_{t_i-1}(i)}{t_i-1} \leq \ra \frac{\ell_{t_i-1}(i-1)}{t_i-1}   \Longleftrightarrow \nonumber \\ & \ell_{t_i}(i) - 1 \leq \ra \ell_{t_i}(i-1) \label{eq:p2:ell:upperbound} \hspace{2ex} \Longleftrightarrow \\ & \frac{\ell_{t_i}(i)}{t_i} \leq \frac{1}{t_i} + \ra \frac{\ell_{t_i}(i-1)}{t_i} \nonumber  \Longleftrightarrow \\ & p_{t_i}(i)\leq \frac{1}{t_i} + \ra p_{t_i}(i-1).\label{eq:p2:probs:upperbound}
\end{align}

At step $t_i+1$ strategy $i$ is not a best response to the opponent's strategy. Then, by Lemma \ref{le:bestresponses}, $i+1$ is the unique best response and so the expected payoff of $i+1$ is better than the expected payoff of $i$ given the probability distributions at step $t_i$, that is,
\begin{align*}
& \beta \sum_{j=1}^{i-2} p_{t_i}(j) + \alpha p_{t_i}(i-1)+p_{t_i}(i) \leq \\ & \beta \sum_{j=1}^{i-2} p_{t_i}(j) + \beta p_{t_i}(i-1)+\alpha p_{t_i}(i).
\end{align*}
Since $\alpha>1$, the above implies that
\begin{equation}\label{eq:p2:probs:lowerbound}
p_{t_i}(i) \geq \ra p_{t_i}(i-1),
\end{equation}
and then that
\begin{equation}\label{eq:p2:ell:lowerbound}
\ell_{t_i}(i) \geq \ra \ell_{t_i}(i-1).
\end{equation}
By definition of $t_i$ action $i$ will not be played anymore until time step $\ts$. Similarly, Lemma \ref{le:bestresponses} shows that $i-1$ will not be a best response twice in the time interval $[1,\ts]$ and so will not be played until step $\ts$. Therefore, the claim follows from (\ref{eq:p2:ell:upperbound}), (\ref{eq:p2:probs:upperbound}), (\ref{eq:p2:probs:lowerbound}) and (\ref{eq:p2:ell:lowerbound}).
\end{proof}

\newcommand{\rb}{\frac{\beta}{\alpha-1}}

\begin{lemma}\label{lemma:p3}
For all strategies $i\in\{3n+1,\ldots,4n-1\}$ and all $t\in \{t_i,\ldots,\ts\}$, it holds:
\[
\ra p_t(i-1) \leq p_t(i) \leq \frac{1}{t} + \ra p_t(i-1) + \rb p_t(i-n)
\]
and therefore,
\[
\ra \ell_t(i-1) \leq \ell_t(i) \leq {1} + \ra \ell_t(i-1)+ \rb \ell_t(i-n).
\]
\end{lemma}
\begin{proof}
By definition of $t_i$, strategy $i$ is played at time step $t_i$. This means that $i$ is a best response for the players after $t_i-1$ steps. In particular, the expected payoff of $i$ is better than the expected payoff of $i+1$, that is,
\begin{align*}
& \beta \left(\sum_{j=1}^{2n} p_{t_i-1}(j) + \sum_{j=i-n}^{i-2} p_{t_i-1}(j)\right) + \alpha p_{t_i-1}(i-1)+ \\ & p_{t_i-1}(i) \geq \beta \left(\sum_{j=1}^{2n} p_{t_i-1}(j) + \sum_{j=i-n+1}^{i-2} p_{t_i-1}(j)\right) + \\ &  \beta p_{t_i-1}(i-1)+\alpha p_{t_i-1}(i).
\end{align*}
Since $\alpha>1$, the above implies that $p_{t_i-1}(i) \leq \ra p_{t_i-1}(i-1) + \rb p_{t_i-1}(i-n)$. Similarly to the proof of Lemma \ref{lemma:p2} above this can be shown to imply
\begin{align}
p_{t_i}(i) & \leq \frac{1}{t_i} + \ra p_{t_i}(i-1) + \rb  p_{t_i}(i-n), \label{eq:p3:probs:upperbound}\\
\ell_{t_i}(i) & \leq 1+ \ra \ell_{t_i}(i-1) + \rb  \ell_{t_i}(i-n). \label{eq:p3:ell:upperbound}
\end{align}

At time step $t_i+1$ strategy $i$ is not a best response to the opponent's strategy. By Lemma \ref{le:bestresponses}, $i+1$ is the unique best response and so the expected payoff of $i+1$ is better than the expected payoff of $i$, thus implying that
\begin{align}
p_{t_i}(i) & \geq \ra p_{t_i}(i-1) +\rb p_{t_i}(i-n) \nonumber \\ & \geq \ra p_{t_i}(i-1), \label{eq:p3:probs:lowerbound}\\
\ell_{t_i}(i) & \geq \ra \ell_{t_i}(i-1) + \rb \ell_{t_i}(i-n) \nonumber \\ & \geq \ra \ell_{t_i}(i-1). \label{eq:p3:ell:lowerbound}
\end{align}
Similarly to Lemma \ref{lemma:p2}, the claim follows from (\ref{eq:p3:probs:upperbound}), (\ref{eq:p3:ell:upperbound}), (\ref{eq:p3:probs:lowerbound}) and (\ref{eq:p3:ell:lowerbound}), the definition of $t_i$ and the fact that, by Lemma \ref{le:bestresponses}, a strategy belonging to $\{3n+1,\ldots,4n-1\}$ is never twice in time a best response in the time interval $[1,\ts]$.
\end{proof}

The next lemma shows that we can ``forget'' about the first $2n$ actions at the cost of paying an exponentially small addend in the payoff function.

\begin{lemma}\label{lemma:forgetthepast}
For any $\delta>0$, $\alpha=1+\frac{1}{n^{\expo}}$ and $\beta = 1-\frac{1}{n^{2(\expo)}}$, $\sum_{j=1}^{2n} p_{\ts}(j) \leq 2^{-n^{\delta}}$.
\end{lemma}
\begin{proof}
We first rewrite and upper bound the sum of the probabilities we are interested in:
\begin{align*}
  \sum_{j=1}^{2n} p_{\ts}(j) &= \sum_{j=1}^{2n} \left[\frac{\ell_{\ts}(j)}{\sum_{j=1}^{4n}\ell_{\ts}(j)}\right] = \frac{\sum_{j=1}^{2n} \ell_{\ts}(j)}{\sum_{j=1}^{4n}\ell_{\ts}(j)} \\ &= \frac{1}{\sum_{j=2n+1}^{4n}\ell_{\ts}(j)} \leq \frac{1}{\ell_{\ts}(4n-1)}.
\end{align*}
Note that by Lemmata \ref{lemma:p1}, \ref{lemma:p2} and \ref{lemma:p3} we have that
\begin{align*}
  \ell_{\ts}(4n-1) & \geq \ra \ell_{\ts}(4n-2) \geq \left(\ra\right)^2 \ell_{\ts}(4n-3)\\ &\geq \left(\ra\right)^{3n-1} \ell_{\ts}(n) = \left(\ra\right)^{3n-1}.
\end{align*}
By plugging in the values of $\alpha$ and $\beta$ given in the hypothesis we have that
\begin{align*}
  \sum_{j=1}^{2n} p_{\ts}(j) & \leq  \frac{1}{\left(\left(1+\frac{1}{n^{\expo}}\right)^{n^{\expo}}\right)^{3n^{\delta}-\frac{1}{n^{\delta}}}} 
  \leq \frac{1}{2^{3n^{\delta}-\frac{1}{n^{\delta}}}}
  \\ &
  \leq \frac{1}{2^{n^{\delta}}},
\end{align*}
where the penultimate inequality follows from the observation that the function $(1+1/x)^x>2$ for $x>2$.
\end{proof}

The theorem below generalizes the above arguments to the cycles that FP visits in the last block of the game, i.e., the block which comprises strategies $S=\{2n+1,\ldots,4n\}$. Since we focus on this part of the game, to ease the presentation, our notation uses circular arithmetic on the elements of $S$. For example, the action $j+2$ will denote action $2n+2$ for $j=4n$ and the action $j-n$ will be the strategy $3n+1$ for $j=2n+1$. Note that under this notation $j-2n=j+2n=j$ for each action $j$ in the block.

\begin{theorem}\label{thm:tailing}
For any $\delta>0$, $\alpha=1+\frac{1}{n^{\expo}}$ and $\beta = 1-\frac{1}{n^{2(\expo)}}$, $n$ sufficiently large, any $t \geq \ts$ we have
\begin{align*}
\frac{p_t(i)}{p_t(i-1)} & \geq 1+\frac{1}{n^{\expo}} \mbox{ for all $i\in S$ with $i \neq s_t, s_t+1$,}\\
\frac{p_t(i)}{p_t(i-1)} & \leq 1+\frac{3}{n^{\expo}} \mbox{ for all $i\in S$}.
\end{align*}
\end{theorem}
\begin{proof}
The proof is by induction on $t$.

\medskip
\noindent{\bf Base.} For the base of the induction, consider $t=\ts$ and note that at that point $s_{\ts}=4n$ and $s_{\ts}+1=2n+1$. Therefore we need to show the lower bound for any strategy $i \in \{2n+2,\ldots, 4n-1\}$. From Lemmata \ref{lemma:p2} and \ref{lemma:p3} we note that for $i \neq 4n, 2n+1$,
\[
\frac{p_{\ts}(i)}{p_{\ts}(i-1)} \geq \ra = 1+\frac{1}{n^{\expo}}.
\]
As for the upper bound, we first consider the case of $i \neq 4n, 2n+1$. Lemma \ref{lemma:p2} implies that for $i=2n+2,\ldots,3n$,
\[
\frac{p_{\ts}(i)}{p_{\ts}(i-1)} \leq \frac{1}{\ts} + \ra,
\]
while Lemma \ref{lemma:p3} implies that for $i=3n+1,\ldots,4n-1$,
\begin{align*}
\frac{p_{\ts}(i)}{p_{\ts}(i-1)} &\leq \frac{1}{\ts} + \ra + \rb \frac{p_{\ts}(i-n)}{p_{\ts}(i-1)} \\ &= \frac{1}{\ts} + \ra + \rb \frac{\ell_{\ts}(i-n)}{\ell_{\ts}(i-1)}.
\end{align*}
To give a unique upper bound for both cases, we only focus on the above (weaker) upper bound and next are going to focus on the ratio $\frac{\ell_{\ts}(i-n)}{\ell_{\ts}(i-1)}$. We use Lemmata \ref{lemma:p2} and \ref{lemma:p3} and get
\begin{align*}
\ell_{\ts}(i-1) & \geq \ra \ell_{\ts}(i-2) \geq \left(\ra\right)^2 \ell_{\ts}(i-3) \\ &\geq \left(\ra\right)^{n-1} \ell_{\ts}(i-n).
\end{align*}
By setting $\alpha$ and $\beta$ as in the hypothesis and noticing that $\ts \geq n \geq n^{\expo}$ we then obtain that
\[
\frac{p_{\ts}(i)}{p_{\ts}(i-1)} \leq 1 + \frac{2}{n^{\expo}} + \left( 1+\frac{1}{n^{\expo}}\right)^{1-n} \frac{n^{2(\expo)}-1}{n^{\expo}}.
\]
We end this part of the proof by showing that the last addend on the right-hand side of the above expression is upper bounded by $\frac{1}{4n^{\expo}}$. To do so we need to prove
\begin{equation}\label{eq:missingbound}
\left(1+\frac{1}{n^{\expo}}\right)^{1-n} \leq \frac{1}{4n^{\expo}}\frac{n^{\expo}}{n^{2(\expo)}-1},\end{equation}
which is equivalent to
\[
\left(\left(1+\frac{1}{n^{\expo}}\right)^{n^{\expo}}\right)^{n^{\delta}-\frac{1}{n^{\expo}}} \geq 4(n^{2(\expo)}-1).
\]
We now lower bound the left-hand side of the latter inequality:
\begin{align*}
\left(\left(1+\frac{1}{n^{\expo}}\right)^{n^{\expo}}\right)^{n^{\delta}-\frac{1}{n^{\expo}}} > \frac{2^{n^\delta}}{2^{\frac{1}{n^{\expo}}}} > \frac{2^{n^\delta}}{2},
\end{align*}
where the first inequality follows from the fact that the function $(1+1/x)^x$ is greater than $2$ for $x > 2$ and the second one follows from the fact that $2^{\frac{1}{n^{\expo}}} < 2$ for $n^{\expo} > 1$. Then, since for $n\geq \sqrt[2(\expo)]{4}$, $5n^{2(\expo)} \geq  4(n^{2(\expo)}-1)$, to prove (\ref{eq:missingbound}) is enough to show
\[
2^{n^{\delta}} \geq 2 (5n^{2(\expo)}) \Longleftrightarrow  n^{\delta} \geq 2(\expo) \log_2(10n).
\]
To prove the latter, since $\delta >0$, it is enough to observe that the function $n^\delta$ is certainly bigger than the function $2\log_2(10n)>2(\expo)\log_2(10n)$ for $n$ large enough (e.g., for $\delta=1/2$, this is true for $n>639$).


The following claim concludes the proof of the base of the induction.
\begin{claim}
The upper bound holds at time step $\ts$ for $i=4n,2n+1$.
\end{claim}
\begin{proof}
We first show the claim for $i=4n$. At time step $\ts$ FP prescribes to play $4n$. This in particular means that the strategy $4n$ achieves a payoff which is at least as much as that of action $2n+1$ after $\ts-1$ time steps. We write down the inequality given by this fact focusing only on the last $2n$ strategies (we will consider the first strategies below) and obtain:
\begin{align}\label{eq:base:payoff1}
& p_{\ts-1}(4n)+\alpha p_{\ts-1}(4n-1)+\beta p_{\ts-1}(3n) \geq \nonumber \\ & \alpha p_{\ts-1}(4n)+p_{\ts-1}(2n+1)+\beta p_{\ts-1}(4n-1)
\end{align}
and then since $\alpha>1$
\begin{align*}
\frac{p_{\ts-1}(4n)}{p_{\ts-1}(4n-1)} \leq& \ra + \frac{\beta}{\alpha-1} \frac{p_{\ts-1}(3n)}{p_{\ts-1}(4n-1)} \\ &- \frac{1}{\alpha-1} \frac{p_{\ts-1}(2n+1)}{p_{\ts-1}(4n-1)}.
\end{align*}
Similarly to the proof of Lemma \ref{lemma:p2} above this can be shown to imply
\begin{align}
\frac{p_{\ts}(4n)}{p_{\ts}(4n-1)} \leq& \frac{1}{\ts}+\ra + \frac{\beta}{\alpha-1} \frac{p_{\ts}(3n)}{p_{\ts}(4n-1)} \nonumber \\ & - \frac{1}{\alpha-1} \frac{p_{\ts}(2n+1)}{p_{\ts}(4n-1)} \nonumber \\ \leq& \frac{1}{\ts}+\ra + \frac{\beta}{\alpha-1} \frac{p_{\ts}(3n)}{p_{\ts}(4n-1)}.\label{eq:newclaim:ub}
\end{align}
We now upper bound the ratio $\frac{\beta}{\alpha-1}\frac{p_{\ts}(3n)}{p_{\ts}(4n-1)}$. By repeatedly using Lemmata \ref{lemma:p2} and \ref{lemma:p3} we have that
\begin{align*}
p_{\ts}(4n-1) &\geq \ra p_{\ts}(4n-2) \geq \left(\ra\right)^2 p_{\ts}(4n-3) \\ &\geq \ldots \geq \left(\ra\right)^{n-1} p_{\ts}(3n).
\end{align*}
This yields
\[
\frac{\beta}{\alpha-1}\frac{p_{\ts}(3n)}{p_{\ts}(4n-1)}\leq \frac{\beta}{\alpha-1} \left(\rainv\right)^{n-1} \leq \frac{1}{4 n^{\expo}},
\]
where the last inequality is proved above (see (\ref{eq:missingbound})). Therefore, since $t \geq n^{\expo}$, (\ref{eq:newclaim:ub}) implies
\[
\frac{p_{\ts}(4n)}{p_{\ts}(4n-1)} \leq 1+\frac{2}{n^{\expo}}+\frac{1}{4n^{\expo}}.
\]
To conclude this part of the proof we must now consider the contribution to (\ref{eq:base:payoff1}) of the actions $1,\ldots, 2n$ that are not in the last block. However, Lemma \ref{lemma:forgetthepast} shows that all those actions are played with probability $1/2^{n^{\delta}}$ at time $\ts$. Thus the overall contribution of these strategies is upper bounded by $\frac{1}{2^{n^{\delta}}} (\alpha-\beta) \leq  \frac{1}{2^{n^{\delta}}}$. Similarly to the above, we observe that, for $n$ sufficiently large, $n^\delta\geq\log_2(4n)\geq (1-\delta) \log_2(4n)$ which implies that $\frac{1}{2^{n^{\delta}}} \leq \frac{1}{4n^{\expo}}$. This concludes the proof of the upper bound at time $\ts$ for $i=4n$.

Consider now the case $i=2n+1$. At time step $\ts+1$, $4n$ is not played by FP, which means that $4n$ is not a best response after $\ts$ time steps. By Lemma \ref{le:bestresponses}, the best response is $2n+1$; then, in particular, the payoff of $2n+1$ is not smaller than the payoff of $4n$ at that time. We write down the inequality given by this fact focusing only on the last $2n$ strategies (we will consider the first strategies below) and obtain
\begin{align*}
&p_{\ts}(4n)+\alpha p_{\ts}(4n-1)+\beta p_{\ts}(3n) \leq \\ &\alpha p_{\ts}(4n)+p_{\ts}(2n+1)+\beta p_{\ts}(4n-1)
\end{align*}
and then since $\alpha>1$
\begin{align}
\frac{p_{\ts}(4n)}{p_{\ts}(4n-1)} \geq & \ra + \frac{\beta}{\alpha-1} \frac{p_{\ts}(3n)}{p_{\ts}(4n-1)} \nonumber \\ &- \frac{1}{\alpha-1} \frac{p_{\ts}(2n+1)}{p_{\ts}(4n-1)}.\label{eq:newclaim:lowerbound}
\end{align}
We next show that $\frac{\beta p_{\ts}(3n) - p_{\ts}(2n+1)}{(\alpha-1)p_{\ts}(4n-1)}\geq$ $- \frac{1}{4n^{\expo}}$ or equivalently that $\frac{p_{\ts}(3n)}{p_{\ts}(2n+1)} \geq \frac{1}{\beta}-\frac{(\alpha-1)p_{\ts}(4n-1)}{4\beta n^{\expo} p_{\ts}(2n+1)}$. To prove this it is enough to show that $\frac{p_{\ts}(3n)}{p_{\ts}(2n+1)} \geq \frac{1}{\beta}$. We observe that
\begin{align*}
\frac{p_{\ts}(3n)}{p_{\ts}(2n+1)} &= \frac{p_{\ts}(3n)}{p_{\ts}(3n-1)}\frac{p_{\ts}(3n-1)}{p_{\ts}(3n-2)}\cdots \frac{p_{\ts}(2n+2)}{p_{\ts}(2n+1)}\\
& \geq \left(\ra\right)^{n-1} \geq \frac{1}{\beta} \nonumber,
\end{align*}
where the first inequality follows from Lemma \ref{lemma:p2} and the second inequality follows from the observation (similar to the above) that for $n$ sufficiently large $n^{\delta}\geq 2 \log_2(2n)$. Then to summarize, for $\alpha$ and $\beta$ as in the hypothesis, (\ref{eq:newclaim:lowerbound}) implies that
\[
\frac{p_{\ts}(4n)}{p_{\ts}(4n-1)} \geq 1 + \frac{1}{n^{\expo}}-\frac{1}{4 n^{\expo}}.
\]
As above we consider actions $1,\ldots, 2n$ and observe that their contribution to the payoffs is upper bounded by $\frac{1}{4n^{\expo}}$. Now to conclude the proof of the claim for the case $i=2n+1$ we simply notice that the above implies $p_{\ts}(4n-1)<p_{\ts}(4n)$ and Lemmata \ref{lemma:p2} and \ref{lemma:p3} imply that $p_{\ts}(2n+1)<p_{\ts}(4n-1)$ which together prove the claim.
\end{proof}

\medskip
\noindent{\bf Inductive step.} Now we assume the claim is true until time step $t-1$ and we show it for time step $t$. By inductive hypothesis, the following is true, with $j \neq s_{t-1}, s_{t-1}+1$
\begin{align}
1+\frac{1}{n^{\expo}} \leq \frac{p_{t-1}(j)}{p_{t-1}(j-1)} & \leq 1+\frac{3}{n^{\expo}}, \label{eq:indstep:indhyp:1}\\
\frac{p_{t-1}(s_{t-1})}{p_{t-1}(s_{t-1}-1)} & \leq 1+\frac{3}{n^{\expo}}, \nonumber \\
\frac{p_{t-1}(s_{t-1}+1)}{p_{t-1}(s_{t-1})} & \leq 1+\frac{3}{n^{\expo}}.\label{eq:indstep:indhyp:2}
\end{align}
We first consider the case in which $s_{t} \neq s_{t-1}$. By Lemma \ref{le:bestresponses}, the strategy played at time $t$ is $s_{t-1}+1$, i.e., $s_t=s_{t-1}+1$. Let $s_{t-1}=i$ and then we have $s_t=i+1$. By inductive hypothesis, for all the actions $j \neq i, i+1, i+2$ we have
\begin{equation}\label{eq:indhyp}
\ra=1+\frac{1}{n^{\expo}} \leq \frac{p_{t}(j)}{p_{t}(j-1)} \leq 1+\frac{3}{n^{\expo}}.
\end{equation}
Indeed, for these actions $j$, $\ell_{t-1}(j)=\ell_t(j)$ and $\ell_{t-1}(j-1)=\ell_t(j-1)$. Therefore the probabilities of $j$ and $j-1$ at time $t$ are simply those at time $t-1$ rescaled by the same amount and the claim follows from (\ref{eq:indstep:indhyp:1}). The upper bound on the ratio $\frac{p_{t}(i+2)}{p_{t}(i+1)}$ easily follows from the upper bound in (\ref{eq:indstep:indhyp:1}) as $\ell_{t-1}(i+2)=\ell_t(i+2)$ and $\ell_{t-1}(i+1)< \ell_{t}(i+1)=\ell_{t-1}(i+1)+1$. However, as $s_t=i+1$ here we need to prove lower and upper bound also for the ratio $\frac{p_{t}(i)}{p_{t}(i-1)}$ and the upper bound for the ratio $\frac{p_{t}(i+1)}{p_{t}(i)}$.

\begin{claim}\label{cla:indstep}
$1+\frac{1}{n^{\expo}} \leq \frac{p_{t}(i)}{p_{t}(i-1)} \leq 1+\frac{3}{n^{\expo}}.$
\end{claim}
\begin{proof}
To prove the claim we first focus on the last block of the game, i.e., the block in which players have strategies in $\{2n+1,\ldots,4n\}$. Recall that our notation uses circular arithmetic on the number of actions of the block.

The fact that action $i+1$ is better than action $i$ after $t-1$ time steps implies that
\begin{align*}
&p_{t-1}(i)+\alpha p_{t-1}(i-1)+\beta p_{t-1}(i-n) \leq \\ &\alpha p_{t-1}(i)+p_{t-1}(i+1)+\beta p_{t-1}(i-1)
\end{align*}
and then since $\alpha>1$
\begin{align}
\frac{p_{t-1}(i)}{p_{t-1}(i-1)} \geq & \ra + \frac{\beta}{\alpha-1} \frac{p_{t-1}(i-n)}{p_{t-1}(i-1)} \nonumber \\ &- \frac{1}{\alpha-1} \frac{p_{t-1}(i-2n+1)}{p_{t-1}(i-1)}.\label{eq:indstep:lowerbound}
\end{align}
We next show that $\frac{\beta p_{t-1}(i-n) - p_{t-1}(i-2n+1)}{(\alpha-1)p_{t-1}(i-1)}\geq - \frac{1}{4n^{\expo}}$ or equivalently that $\frac{p_{t-1}(i-n)}{p_{t-1}(i-2n+1)} \geq \frac{1}{\beta}-\frac{(\alpha-1)p_{t-1}(i-1)}{4\beta n^{\expo} p_{t-1}(i-2n+1)}$. To prove this it is enough to show that $\frac{p_{t-1}(i-n)}{p_{t-1}(i-2n+1)} \geq \frac{1}{\beta}$. We observe that
\begin{align*}
\frac{p_{t-1}(i-n)}{p_{t-1}(i-2n+1)} &= \frac{p_{t-1}(i-n)}{p_{t-1}(i-n-1)}\cdots\frac{p_{t-1}(i-2n+2)}{p_{t-1}(i-2n+1)}\\
& \geq \left(\ra\right)^{n-1} \geq \frac{1}{\beta} \nonumber,
\end{align*}
where the first inequality follows from inductive hypothesis 
(we can use the inductive hypothesis as all the actions involved above are different from $i$ and $i+1$) and the second inequality follows from the aforementioned observation that, for sufficiently large $n$, $n^{\delta}\geq 2 \log_2(2n)$. Then to summarize, for $\alpha$ and $\beta$ as in the hypothesis, (\ref{eq:indstep:lowerbound}) implies that
\[
\frac{p_{t}(i)}{p_{t}(i-1)} = \frac{p_{t-1}(i)}{p_{t-1}(i-1)} \geq 1 + \frac{1}{n^{\expo}}-\frac{1}{4 n^{\expo}},
\]
where the first equality follows from $\ell_{t-1}(i)=\ell_{t}(i)$ and $\ell_{t-1}(i-1)=\ell_{t}(i-1)$, which are true because $s_t=i+1$.

Since action $i+1$ is worse than strategy $i$ at time step $t-1$ we have that
\begin{align*}
&p_{t-1}(i)+\alpha p_{t-1}(i-1)+\beta p_{t-1}(i-n) \geq \\ &\alpha p_{t-1}(i)+p_{t-1}(i+1)+\beta p_{t-1}(i-1)
\end{align*}
and then since $\alpha>1$
\begin{align*}
\frac{p_{t-1}(i)}{p_{t-1}(i-1)} \leq & \ra + \frac{\beta}{\alpha-1} \frac{p_{t-1}(i-n)}{p_{t-1}(i-1)} \\ &- \frac{1}{\alpha-1} \frac{p_{t-1}(i-2n+1)}{p_{t-1}(i-1)}.
\end{align*}
Similarly to the proof of Lemma \ref{lemma:p2} above this can be shown to imply
\begin{align}
\frac{p_{t}(i)}{p_{t}(i-1)} \leq & \frac{1}{t}+\ra + \frac{\beta}{\alpha-1} \frac{p_{t}(i-n)}{p_{t}(i-1)} \nonumber \\ &- \frac{1}{\alpha-1} \frac{p_{t}(i-2n+1)}{p_{t}(i-1)} \nonumber \\ \leq & \frac{1}{t}+\ra + \frac{\beta}{\alpha-1} \frac{p_{t}(i-n)}{p_{t}(i-1)}.\label{eq:indstep:ub}
\end{align}
We now upper bound the ratio $\frac{\beta}{\alpha-1}\frac{p_{t}(i-n)}{p_{t}(i-1)}$. By repeatedly using the inductive hypothesis (\ref{eq:indhyp}) we have that
\begin{align*}
p_t(i-1) \geq & \ra p_t(i-2) \geq \left(\ra\right)^2 p_t(i-3) \\ \geq & \left(\ra\right)^{n-1} p_t(i-n).
\end{align*}
(Note again that we can use the inductive hypothesis as none of the actions above is $i$ or $i+1$.) This yields
\[
\frac{\beta}{\alpha-1}\frac{p_{t}(i-n)}{p_{t}(i-1)}\leq \frac{\beta}{\alpha-1} \left(\rainv\right)^{n-1} \leq \frac{1}{4 n^{\expo}},
\]
where the last inequality is proved above (see (\ref{eq:missingbound})). Therefore, since $t \geq n^{\expo}$, (\ref{eq:indstep:ub}) implies the following
\[
\frac{p_{t}(i)}{p_{t}(i-1)} \leq 1+\frac{2}{n^{\expo}}+\frac{1}{4 n^{\expo}}.
\]

To conclude the proof we must now consider the contribution to the payoffs of the actions $1,\ldots, 2n$ that are not in the last block. However, Lemma \ref{lemma:forgetthepast} shows that all those actions are played with probability $1/2^{n^{\delta}}$ at time $\ts$. Since we prove above (see Lemma \ref{le:bestresponses}) that these actions are not played anymore after time step $\ts$ this implies that $\sum_{j=1}^{2n} p_{t}(j) \leq \sum_{j=1}^{2n} p_{\ts}(j) \leq 2^{-n^{\delta}}$. Thus the overall contribution of these strategies is upper bounded by $\frac{1}{2^{n^{\delta}}} (\alpha-\beta) \leq  \frac{1}{2^{n^{\delta}}} \leq \frac{1}{4n^{\expo}}$ where the last bound follows from the aforementioned fact that, for $n$ sufficiently large, $n^\delta\geq(1-\delta)\log_2(4n)$. This concludes the proof of this claim.
\end{proof}

\begin{claim}
$\frac{p_{t}(i+1)}{p_{t}(i)} \leq 1+\frac{3}{n^{\expo}}.$
\end{claim}
\begin{proof}
From (\ref{eq:indstep:lowerbound}) (and subsequent arguments) we get $p_t(i)>p_t(i-1)$ and from (\ref{eq:indhyp}) we get $p_t(i-1) > p_t(i-2) > \ldots > p_t(i-2n+1)=p_t(i+1)$. Therefore, $p_t(i)>p_t(i+1)$ thus proving the upper bound.
\end{proof}

Finally, we consider the case in which $s_{t-1}=s_t$. In this case, for the actions $j \neq s_t,s_t+1$ it holds $\ell_{t-1}(j)=\ell_t(j)$ and $\ell_{t-1}(j-1)=\ell_t(j-1)$. Therefore, similarly to the above, for these actions $j$ the claim follows from (\ref{eq:indstep:indhyp:1}). The upper bound for the ratio $\frac{p_{t}(s_{t}+1)}{p_{t}(s_{t})}$ easily follows from
(\ref{eq:indstep:indhyp:2}) as $\ell_{t-1}(s_t+1)=\ell_t(s_t+1)$ and $\ell_{t-1}(s_t)< \ell_{t}(s_t)=\ell_{t-1}(s_t)+1$. The remaining case to analyze is the upper bound on the ratio $\frac{p_{t}(s_{t})}{p_{t}(s_{t}-1)}$. To prove this we can use \emph{mutatis mutandis} the proof of the upper bound contained in Claim \ref{cla:indstep} with $s_t=i$.
\end{proof}

The claimed performance of Fictitious Play, in terms of the approximation to the best
response that it computes, follows directly from this theorem.

\begin{theorem}\label{thm:final}
For any value of $\delta>0$ and any time step $t$, Fictitious Play returns an $\epsilon$-NE with $\epsilon \geq \frac{1}{2}-O\left(\frac{1}{n^{\expo}}\right)$.
\end{theorem}

\begin{proof}
For $t\leq n$ the result follows since the game is similar to~\cite{Con}. In details, for $t \leq n$ the payoff associated to the best response, which in this case is $s_t+1$, is upper bounded by $1$. On the other hand, the payoff associated to the current strategy prescribed by FP is lower bounded by $\frac{\beta}{i^2} \sum_{j=0}^{i-1} j$ where $i=s_t$. Therefore, the regret of either player normalized to the $[0,1]$ interval satisfies: $\epsilon \geq \frac{1}{\alpha}-\frac{\beta}{\alpha}\frac{i-1}{2i}$.
Since $\frac{i-1}{2i}<1/2$, the fact that $1-\frac{\beta}{2}-\frac{\alpha}{2}+\frac{\alpha}{n^{\expo}}\geq 0$ (which is true given the values of $\alpha$ and $\beta$) yields the claim.
For $t\leq \ts$ the result follows from Lemmata~\ref{lemma:p2} and \ref{lemma:p3};
while the current strategy $s_t$ (for $t\leq t^*$) has payoff approximately 1,
the players' mixed strategies have nearly all their probability on the recently played strategies, but with no pure
strategy having very high probability, so that some player is likely to receive zero payoff;
by symmetry each player has payoff approximately $\frac{1}{2}$. This is made precise below,
where it is applied in more detail to the case of $t>\ts$.

We now focus on the case $t > \ts$. Recall 
that for a set of strategies $S$, $p_t(S)=\sum_{i\in S} p_t(i)$. Let $S_t$ be the set $\{2n+1,\ldots,s_t\}\cup\{s_t+n,\ldots,4n\}$ if $s_t\leq 3n$,
or the set $\{s_t-n,\ldots,s_t\}$ in the case that $s_t>3n$.
Let $S'_t=\{2n+1,\ldots,4n\}\setminus S_t$. Also, let $s^{\max}_t=\arg\max_{i\in\{2n+1,\ldots,4n\}}(p_t(i))$; note that by Theorem~\ref{thm:tailing},
$s^{\max}_t$ is equal to either $s_t$ or $s^-_t$, where $s^-_t=s_t-1$ if $s_t>2n$, or $4n$ if $s_t=2n$.

We start by establishing the following claim:

\begin{claim}\label{cla:mainthm:1}
For sufficiently large $n$, $p_t(S_t)\geq 1-\frac{2n-1}{2^{n^{\delta}}}$.
\end{claim}
\begin{proof}
To see this, note that for all $x\in S'_t$, by $p_t(s^{\max}_t) \geq p_t(s^{\max}_t-1)$ and Theorem~\ref{thm:tailing} we have
\begin{align*}
  \frac{p_t(s^{\max}_t)}{p_t(x)}  & = \frac{p_t(s^{\max}_t)}{p_t(s^{\max}_t-1)}\frac{p_t(s^{\max}_t-1)}{p_t(s^{\max}_t-2)}\ldots\frac{p_t(x+1)}{p_t(x)} \nonumber \\ & \geq \left(1+\frac{1}{n^{\expo}}\right)^{k-1},
\end{align*}
where $k$ is the number of factors on the right-hand side of the equality above, i.e., the number of strategies between $x$ and $s^{\max}_t$. Thus, as $k \geq n$,
\begin{align*}
p_t(x)
& \leq \frac{p_t(s^{\max}_t)}{\left(1+\frac{1}{n^{\expo}}\right)^{k-1}}
\leq \left(1+\frac{1}{n^{\expo}}\right)^{1-k} \\ &\leq \left(1+\frac{1}{n^{\expo}}\right)^{1-n}
\leq 4^{(1-n)/(n^{\expo})}.
\end{align*}
Hence $p_t(S'_t)\leq (2n)4^{(1-n)/(n^{\expo})}=\frac{n 4^{1/n^{\expo}}}{2^{n^{\delta}}}<\frac{2n}{2^{n^\delta}}$, where the last inequality follows from the fact that, for large $n$, $4^{1/n^{\expo}}<2$. Then $p_t(S_t)\geq 1-p_t(S'_t)-p_t(\{1,\ldots,2n\})$,
which establishes the claim, since Lemma~\ref{lemma:forgetthepast} establishes a strong
enough upper bound on $p_t(\{1,\ldots,2n\})$.
\end{proof}
\begin{claim}
$s_t$, the current best response at time $t$, has payoff at least $\beta \left(1-\frac{2n-1}{2^{n^{\delta}}}\right)$.
\end{claim}
\begin{proof}
$s_t$ receive a payoff of at least $\beta$ when the opponent plays any strategy from $S_t$; the claim
follows using Claim \ref{cla:mainthm:1}.
\end{proof}
Let $E_t$ denote the expected payoff to either player that would result if they both select
a strategy from the mixed distribution that allocates to each strategy $x$, the probability $p_t(x)$.
The result will follow from the following claim:
\begin{claim}
For sufficiently large $n$, $E_t\leq \frac{\alpha}{2}+\frac{6}{n^{\expo}}+\alpha \frac{2n}{2n^\delta}$.
\end{claim}
\begin{proof}
The contribution to $E_t$ from strategies in $\{1,\ldots,n\}$, together with strategies in $S'_t$,
may be upper-bounded by $\alpha$ times the probability that any of that strategies get
played. This probability is by Lemma \ref{lemma:forgetthepast} and Claim \ref{cla:mainthm:1} exponentially small, namely $2n/2^{n^\delta}$.

Suppose instead that both players play from $S_t$. If they play different strategies, their
total payoff will be at most $\alpha$, since one player receives payoff 0. If they play the
same strategy, they both receive payoff 1. We continue by upper-bounding the probability that
they both play the same strategy. This is upper-bounded by the largest probability assigned
to any single strategy, namely $p_t(s^{\max}_t)$.

Suppose for contradiction that $p_t(s^{\max}_t) > 6/n^{\expo}$. At this point, note that by Theorem \ref{thm:tailing}, for any strategy $s \in S_t$, we have
\[
\frac{p_t(s^{\max}_t)}{p_t(s)} \leq \left(1+\frac{3}{n^{\expo}}\right)^k,
\]
where $k$ is the distance between $s$ and $s^{\max}_t$. Therefore, denoting $r=\left(1+\frac{3}{n^{\expo}}\right)^{-1}$, we obtain
\begin{align*}
p_t(S_t) &=\sum_{s \in S_t} p_t(s) = p_t(s_t) + \sum_{i=s_t-n}^{s_t-1} p_t(i) \\&\geq p_t(s^{\max}_t) \sum_{k=0}^{n-1} r^k.
\end{align*}
Applying the standard formula for the partial sum of a geometric series we have
\[
p_t(S_t) \geq \frac{6}{n^{\expo}}\left(\frac{1-r^{n}}{1-r} \right)
\]
Noting that $1-r^{n} >\frac{1}{2}$ we have $p_t(S_t) > \frac{6}{n^{\expo}}\cdot(\frac{1}{2})\cdot(\frac{n^{\expo}}{3})$
which is greater than 1, a contradiction.

The expected payoff $E_t$ to either player, is, by symmetry, half the expected total payoff,
so we have $E_t \leq (1-\frac{2n}{2^{n^\delta}}-\frac{6}{n^{\expo}}) \frac{\alpha}{2}  + \frac{6}{n^{\expo}} + \frac{2n}{2n^\delta} \alpha$ which yields the claim.
\end{proof}
We now show that Fictitious Play never achieves an $\epsilon$-value better than $\frac{1}{2}-O\left(\frac{1}{n^{\expo}}\right)$. From the last two claims the regret of either player normalized to $[0,1]$ is
\begin{align*}
\epsilon  \geq & \frac{\beta}{\alpha}\left(1-\frac{2n-1}{2^{n^{\delta}}}\right)-\frac{1}{2}-\frac{6}{\alpha n^{\expo}}-\frac{2n}{2^{n^\delta}} \\ = & \left(1- \frac{n^{\expo}+1}{n^{2(\expo)}+n^{\expo}}\right)\left(1-\frac{2n-1}{2^{n^{\delta}}}\right)\\ &-\frac{1}{2}-\frac{6}{n^{\expo}+1}-\frac{2n}{2^{n^\delta}}
\\ =& \frac{1}{2} - \frac{n^{\expo}+1}{n^{2(\expo)}+n^{\expo}} - \frac{n^{\expo}+1}{n^{2(\expo)}+n^{\expo}}\frac{2n-1}{2^{n^{\delta}}} \\ &-\frac{6}{n^{\expo}+1}-\frac{2n}{2^{n^\delta}}\\
=&\frac{1}{2} - O\left(\frac{1}{n^{\expo}}\right).
\end{align*}
This concludes the proof.
\end{proof}

\section{Upper bound}\label{sec:UB}

In this section, $n$ denotes the number of pure strategies of both players.
Let $a, b$ denote the FP sequences of pure strategies of length $t$,
for players 1 and 2 respectively.
Let $a_{[k:\ell]}$ denote the subsequence $a_k,\ldots, a_\ell$.
We overload notation and use $a$ to also denote the mixed strategy that is uniform on the corresponding sequence.

Let $m^*$ be a best response against $b$, and let $\epsilon$ denote the 
smallest $\epsilon$ for which $a$ is an $\epsilon$-best-response against~$b$. 
To derive a bound on $\epsilon$, we use the most recent occurence of pure strategy 
in $a$. 
For $k \in \{1,\ldots,t\}$, let $f(k)$ denote the last occurrence of $a_k$ in the sequence $a$, that is,
$$
f(k) := \max_{\ell \in \{1,\ldots, t \} ,\ a_\ell = a_k} \ell.
$$
We have the following.

\begin{align}
\nonumber \epsilon &= u_1(m^*, b) - u_1(a,b)\\
\nonumber &=\frac1t \sum_{i=1}^t \big( u_1(m^*,b)-u_1(a_i,b) \big) \\
\nonumber &=\frac1t \sum_{i=1}^t \big[ \frac{f(i)-1}t
(u_1(m^*,b_{[1:f(i)-1]})-u_1(a_i,b_{[1:f(i)-1]})) + \frac{t-f(i)+1}t
(u_1(m^*,b_{[f(i):t]})-u_1(a_i,b_{[f(i):t]}))\big]\\
\label{e:dropbr}
&\leq\frac1t \sum_{i=1}^t \big[ \frac{t-f(i)+1}t
(u_1(m^*,b_{[f(i):t]})-u_1(a_i,b_{[f(i):t]})) \big]\\
&
\label{e:sub1}
\leq\frac1t \sum_{i=1}^t \frac{t-f(i)+1}t\\
&=
\label{e:msbound}
1 + \frac1t - \frac1{t^2} \sum_{i=1}^t f(i)
\end{align}
Inequality \eqref{e:dropbr} holds since $a_i$ is a best response against 
$b_{[1:f(i)-1]}$, by definition.
Inequality \eqref{e:sub1} holds since payoffs are in the range $[0,1]$.
To provide a guarantee on the performance of FP, we find the sequence~$a$ 
that maximizes the RHS of \eqref{e:msbound}, i.e., that minimizes $\sum_{i=1}^t f(i)$.
\begin{definition}
\label{e:sumf}
For a FP sequence $a$, let $S(a) := \sum_{i=1}^t f(a_i)$ and let $\bs= \arg \min_a S(a)$.
\end{definition}

The following three lemmata allow us to characterize \bs, the sequence that minimizes $S(a)$.

\begin{lemma}
\label{l:ndistinct}
The entries of $\bs$ take on exactly $n$ distinct values.
\end{lemma}

\begin{proof}
The entries of an FP sequence can take on at most $n$ distinct values.
Suppose for the sake of contradiction that the entries of \bs take on
strictly less than $n$ distinct values. 
Then there is a pure strategy, say $m$, that does not appear in \bs and
a pure strategy $m'$ that appears more than $t/n$ times.
Obtain $a$ from $\bs$ by replacing a single occurrence of $m'$ in $\bs$ with $m$.
Then $S(a)<S(\bs)$, a contradiction.
\end{proof}

We now define a transformation of an FP sequence $a$ into a new sequence $a'$
so that $S(a') < S(a)$ if $a \ne a'$.

\begin{definition}
\label{d:transform}
Suppose the entries of $a$ take on $d$ distinct values.
We define $x_1,\ldots,x_d$ to be the last occurrences, $\{f(a_i)\ |\ i\in[t]\}$,
in ascending order.
Formally, let $x_d:=a_t$ and for $k<d$ let $x_{k}:=a_i$ be such that
$$
i := \arg\max_{j=1,\ldots,t} a_j \notin \{x_{k+1},\ldots,x_d\}.
$$
For $i=1,\ldots,d$, let
$$
\#(x_i):=|\{a_j\ |\ j\in[t], a_j=x_i\}|,
$$
which is the number of occurrences of $x_i$ in $a$.
Define $a'$ as
$$
a' :=
\underbrace{x_1,\ldots,x_1}_{\#(x_1)}, \underbrace{x_2,\ldots,x_2}_{\#(x_2)}, \cdots,\underbrace{x_d,\ldots,x_d}_{\#(x_d)}.
$$
\end{definition}

\begin{lemma}
\label{l:transform}
For any FP sequence $a$, let $a'$ be as in Definition~\ref{d:transform}.
If $a' \ne a$ then $S(a') < S(a)$.
\end{lemma}
%
\begin{proof}
For all $i=1,\ldots,t$ we have $f(a'_i)\le f(a_i)$, and since $a'\ne a$ there is at least one $i$ such that $f(a'_i)< f(a_i)$.
\end{proof}

\begin{lemma}
\label{l:equal}
Let $n,t\in \mathbb{N}$ be such that $n | t$. Let $a$ be a sequence of length $t$ of the form
$$
a
=
\underbrace{1,\ldots,1}_{c_1}, \underbrace{2,\ldots,2}_{c_2}, \cdots,\underbrace{n,\ldots,n}_{c_n}.
$$
Then $S(a)$ is minimized if and only if $c_1=\cdots=c_n=t/n$.
\end{lemma}

\begin{proof}
We refer to the maximal length subsequence of entries with value $u \in \{1,\ldots n\}$ as \emph{block} $u$.
Consider two adjacent blocks $u$ and $u+1$, where block $u$ starts
at $i$ and block $u+1$ starts at $j$ and finishes at $k$.
The contribution of these two blocks to $S(a)$ is
$$
\sum_i^{j-1} (j-1) + \sum_j^k k = j^2 - (k+i)j + (i+k)\ .
$$
If $k+i$ is even, this contribution is minimized when $j=\frac{k+i}{2}$.
If $k+i$ is odd, this contribution is minimized for both values $j= \lfloor \frac{k+i}{2} \rfloor$ and $j=\lceil \frac{k+i}{2} \rceil$.

Now suppose for the sake of contradiction that $S(a)$ is minimized when $c_1=\cdots=c_n=t/n$
does not hold.
There are two possibilities.
Either there are two adjacent blocks whose lengths differ by more than one, in
which case we immediately have a contradiction.
If not, then it must be the case that all pairs of adjacent blocks differ in length
by at most one.
In particular, there must be a block of length $t/n+1$ and another of length $t/n-1$
with all blocks in between of length $t/n$.
Flipping the leftmost of these blocks with its right neighbor will not change the sum $S(a)$.
Repeatedly doing this until the blocks of lengths $t/n+1$ and $t/n-1$ are adjacent,
does not change $S(a)$.
Then we have two adjacent blocks that differ in length by more than one, which contradicts
the fact that $S(a)$ was minimized.
\end{proof}

\begin{theorem}
If $n | t$, the FP strategies $(a,b)$ are an $\epsilon^*$-equilibrium,
where $$\epsilon^* = \frac{1}{2}+\frac{1}{t}-\frac{1}{2n}\ .$$
\end{theorem}

\begin{proof}
By symmetry, it suffices to show that \bs is an $\epsilon^*$-best-response against $b$.
Applying Lemma~\ref{l:ndistinct}, Lemma~\ref{l:transform} and Lemma~\ref{l:equal},
we have that
$$
\bs =
\underbrace{m_1,\ldots,m_1}_{t/n}, \underbrace{m_2,\ldots,m_2}_{t/n}, \cdots,\underbrace{m_n,\ldots,m_n}_{t/n},
$$
where $m_1,\ldots,m_n$ is an arbitrary labeling of player 1's pure strategies.
Using~\eqref{e:msbound}, we have that
\begin{align*}
\epsilon \le & \quad 1 + \frac1t - \frac1{t^2} \sum_{i=1}^t f(\bs_i) \\
=& \quad 1 + \frac1t - \frac1{t^2}  \frac{t}{n} \sum_{i=1}^n (\frac{i\cdot t}{n}) \\
=& \quad 1 + \frac1t - \frac{n+1}{2n} \\
=& \quad \frac12 + \frac1t - \frac{1}{2n} \\
\end{align*}
This concludes the proof.
\end{proof}
For $t$ superlinear in $n$, we asymptotically achieve a $(\frac{1}{2}-\frac{1}{2n})$-Nash equilibrium.
\section{Discussion}

Daskalakis et al.~\cite{DMP} gave a very simple algorithm that achieves an approximation
guarantee of $\frac{1}{2}$; subsequent algorithms e.g.~\cite{BBM,TS} improved on this, but at the expense
of being more complex and centralized, commonly solving one or more derived LPs
from the game. Our result suggests that further work on the topic might address the
question of whether $\frac{1}{2}$ is a fundamental limit to the approximation performance
obtainable by certain types of algorithms that are in some sense simple or decentralized.
The question of specifying appropriate classes of algorithms is itself challenging,
and is also considered in~\cite{DFPPV} in the context of algorithms that provably fail
to find Nash equilibria without computational complexity theoretic assumptions.

\end{document}